\documentclass{article}
\usepackage{fix-cm}
\usepackage{flexvdr_preprint,times}
\usepackage{amsmath,amssymb,amsthm,booktabs,graphicx,microtype,xcolor,xspace,url}
\usepackage{algorithm,algorithmic,colortbl,array,tikz}
\usepackage{placeins,pifont,needspace}
\definecolor{datasetgray}{gray}{0.92}
\usepackage{hyperref}
\definecolor{routepurple}{HTML}{7038AD}
\hypersetup{colorlinks=true,citecolor=cyan,linkcolor=routepurple,urlcolor=routepurple,
pdftitle={FlexVDR: Learning to Allocate Matryoshka Representation Budgets for Visual Document Retrieval},pdfauthor={Haowen Xiang, Yibo Yan, Xin Zou, Shuliang Liu, Yi Cao, Mingdong Ou, Xuming Hu}}
\newcommand{\method}{\textsc{FlexVDR}\xspace}
\newcommand{\R}{\mathcal R}
\newcommand{\E}{\mathbb E}
\definecolor{algcommentcolor}{HTML}{008080}
\newcommand{\algstage}[1]{\item[]\textcolor{algcommentcolor}{\footnotesize\ttfamily /* #1 */}}
\newcommand{\algcomment}[1]{\hfill\textcolor{algcommentcolor}{\footnotesize\ttfamily // #1}}

\newtheorem{proposition}{Proposition}
\title{FlexVDR: Learning to Allocate\\Matryoshka Representation Budgets \\for Visual Document Retrieval}
\author{\parbox[t]{\dimexpr\textwidth-2\tabcolsep\relax}{\centering\normalfont
Haowen Xiang$^{1}$, Yibo Yan$^{1,2,3}$, Xin Zou$^{1,3}$, Shuliang Liu$^{1}$\\[3pt]
Yi Cao$^{2}$, Mingdong Ou$^{2}$, Xuming Hu$^{1,3}$\\[7pt]
{\small
$^{1}$\,The Hong Kong University of Science and Technology (Guangzhou)\\
$^{2}$\,Alibaba Group\\
$^{3}$\,The Hong Kong University of Science and Technology\endgraf}
}}
\begin{document}
\maketitle
\begin{abstract}
Visual document retrieval relies on rich multi-vector representations to capture fine-grained evidence in document pages, but encoding and matching these representations can be computationally expensive. Existing compression and elastic representation methods expose smaller configurations, yet making multiple budgets available does not determine which configuration a query should use. The allocation decision must estimate the retrieval value of alternative representations before incurring their full computation, while accommodating different resource preferences. We introduce \method, a framework for query-adaptive budget allocation within a single retriever. We construct a joint family of routes over encoder depth, embedding width, and token count. Retrieval supervision and distillation train the compact routes, while a native bypass preserves the full-budget reference. A Budget Allocator predicts route-specific retrieval utility from early query features and combines these predictions with a configurable resource preference to select a route. Experiments across five backbones on ViDoRe V1, V2, and V3 public show that \method retains much of the native retrieval quality while selecting compact routes for individual queries. A configurable resource preference allows the same retriever to adjust its computation without retraining or switching models.
\end{abstract}
\section{Introduction}
Visual document retrieval (VDR) retrieves pages from their rendered appearance, preserving evidence in tables, figures, and layout that can be difficult to recover through text extraction. Multi-vector retrievers such as ColPali \citep{faysse2025colpali} combine a vision--language encoder with token-level late interaction to capture fine-grained query--page correspondences. With document representations prepared offline, online retrieval still requires query encoding and vector-level matching against the stored representations \citep{khattab2020colbert}. The resulting challenge is to reduce this computation while retaining the information needed for accurate retrieval.

Compression and elastic representation learning make smaller representations available: Matryoshka embeddings support nested widths of the final-layer representation \citep{kusupati2022matryoshka}, while depth-aware extensions make intermediate-layer representations available for retrieval (Figure~\ref{fig:overview}(a,b)) \citep{wang20252d,xiang2026mmmatryoshka}. Token pruning and pooling reduce the number of document vectors involved in matching \citep{yan2025docpruner,roy2026colsnap}. Together, these approaches provide choices over encoder depth, embedding width, and token count.

\begin{figure}[t]
\centering
\includegraphics[width=\linewidth]{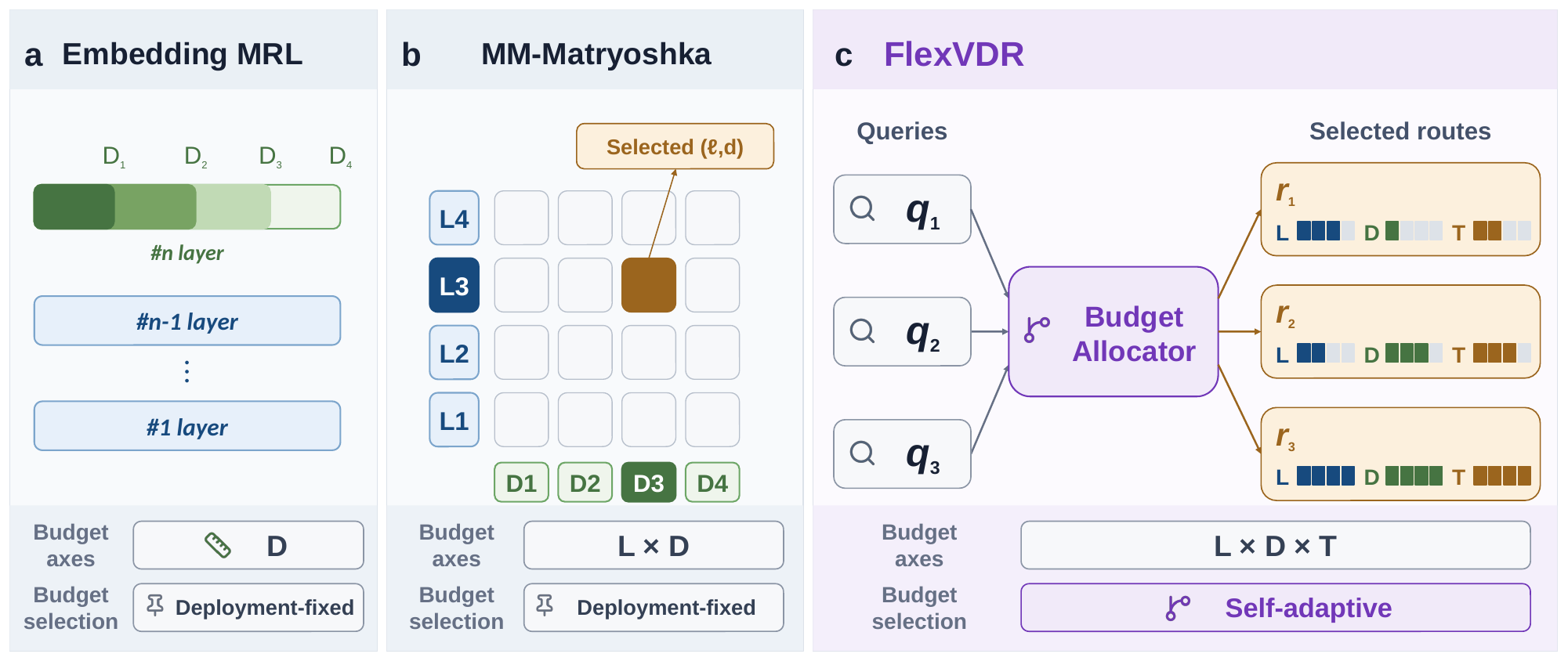}
\caption{\textbf{Representation choices and budget selection.} (a) Embedding MRL provides nested widths at the final layer; (b) MM-Matryoshka provides depth--width choices. Both illustrated settings use one fixed configuration across queries. (c) \method selects depth, width, and query/document token budgets per query. $L$, $D$, and $T$ label the depth, width, and token-count axes.}
\label{fig:overview}
\end{figure}

The evidence that makes a page relevant is query-specific: queries may target information in text, tables, or figures, and late-interaction matching associates individual query terms with different page regions \citep{faysse2025colpali}. A deployment-fixed configuration, however, assigns the same depth, width, and token count to every query. Choosing such a configuration sets a common compromise between representation capacity and computation across the query workload. Adding more configurations broadens the available choices, but a fixed selection still leaves the incoming query without a role in deciding which representation to use. Learning useful compact representations and deciding when to use each one therefore remain distinct parts of the retrieval problem. This motivates our central question: \emph{how can a single retriever allocate representation budgets to individual queries while accommodating different resource preferences?}

Query-adaptive allocation poses two design requirements. Dynamic early-exit models provide a precedent by using intermediate predictions and adjustable inference-time thresholds to control encoder computation \citep{xin2020deebert,liu2020fastbert}. \ding{182} The allocation choice must be made before executing the remaining encoder layers, using query information available at an early stage. Running every candidate route before choosing one would already incur the computation that allocation is intended to avoid. The decision must therefore anticipate the retrieval value of query representations that have not yet been produced. \ding{183} Changing the desired computation level should not require learning a new allocation rule. A query may be processed under different resource preferences, so its representation choice should remain adjustable after training.

To address these requirements, we introduce \method, a framework for query-adaptive budget allocation within a single retriever (Figure~\ref{fig:overview}(c)). Its Budget Allocator predicts the retrieval utility of candidate routes from early query features and combines these estimates with a resource preference to select a route. The utility estimates describe what each representation is expected to offer for the current query, while the resource preference determines how strongly selection favors compact configurations. Learning the utilities independently of this preference allows the same predictions to be reused when the desired quality--computation trade-off changes. The candidate routes jointly vary encoder depth, embedding width, and query/document token count, with learned compact representations and the native full route providing the available choices. The resulting retriever can adjust its computation without retraining or switching models.

Our contributions are as follows:
\begin{itemize}
    \item[\ding{182}] \textbf{Query-adaptive budget allocation.} We develop a Budget Allocator that predicts route-specific retrieval utility from early query features and uses a separately specified resource preference to select a representation budget.
    \item[\ding{183}] \textbf{A joint representation family for allocation.} We integrate selectable encoder depth, embedding width, and query/document token count into a common route family, with compact representations and the native full route as allocation choices.
    \item[\ding{184}] \textbf{Evaluation across backbones and benchmarks.} We evaluate five backbones on ViDoRe V1, V2, and V3 public, and examine computation across resource preferences.
\end{itemize}
\section{Related work}
\paragraph{Multimodal multi-vector retrieval.}
Multi-vector retrieval preserves fine-grained evidence through token-level matching. ColBERT established contextualized late interaction for text retrieval \citep{khattab2020colbert}; ColPali transfers this approach to page images \citep{faysse2025colpali}. Subsequent released models expand backbone capacity and retrieval training, including Nemotron ColEmbed V2 and webAI-ColVec \citep{moreira2026nemotron,webai2026colvec}. ViDoRe V1 \citep{faysse2025colpali} evaluates this setting across visually rich domains, while ViDoRe V2 \citep{mace2025vidore} introduces more demanding query construction and relevance judgments. ViDoRe V3 public \citep{loison2026vidorev3} extends evaluation to multimodal retrieval-augmented generation with diverse query types and source-grounding annotations. These retrievers and benchmarks establish the visual retrieval setting in which we study how much representation capacity to allocate to each query.
\paragraph{Matryoshka representation learning.}
Matryoshka representation learning concentrates useful information in nested dimension prefixes \citep{kusupati2022matryoshka}. Two-dimensional training additionally uses encoder depth \citep{wang20252d}; MM-Matryoshka studies depth--dimension elasticity for multi-vector VDR \citep{xiang2026mmmatryoshka}. Complementary compression methods reduce the storage or number of document vectors. ColBERTv2 reduces storage through residual compression, representing token vectors as centroid identifiers and quantized residuals \citep{santhanam2022colbertv2}. DocPruner reduces page embeddings using document-dependent patch importance \citep{yan2025docpruner}. ColSNAP learns spatially pooled compression tiers that can be chosen at indexing time \citep{roy2026colsnap}. \citet{qin2026multivectorcompression} study query-agnostic document compression under a fixed vector budget, using attention-guided clustering to select centroids and aggregate token representations. Together, nested representations and vector compression provide ways to vary representation budgets. We build a joint depth--width--token family and learn to select its budget for each incoming query.
\paragraph{Budget-oriented search.}
Input-dependent computation controls how much encoding or search work an input requires. Dynamic early-exit methods such as DeeBERT and FastBERT use intermediate prediction confidence to decide whether encoder computation should continue, with thresholds adjustable after training \citep{xin2020deebert,liu2020fastbert}. In retrieval, AdANNS uses nested representations to construct adaptive approximate-nearest-neighbor search pipelines \citep{rege2023adanns}. Within late interaction, Argus conditions region-level expert routing on the query and cached document features to construct query-dependent document representations \citep{abdallah2026argus}. \citet{rishi2026queryaware} formulate query-aware document-token selection for second-stage reranking as a MaxSim coverage problem under a per-page token budget. Col-Bandit reduces query-time scoring work by progressively evaluating MaxSim contributions and using confidence bounds to eliminate uncompetitive candidates \citep{pony2026colbandit}. The allocation decision in \method concerns the joint encoder depth, embedding width, and query/document token count. It predicts route-specific retrieval utilities from shared early query features and selects one route before the remaining query-encoding and matching computation. Learning these utilities separately from the resource preference allows the same fitted predictor to support different choices within this joint budget space.
\section{\method: Learning and allocating representation budgets}
\label{sec:method}
\paragraph{Overview.}
\method first learns selectable representations around a frozen native retriever, then fixes this family and trains a Budget Allocator to predict route utilities (Figure~\ref{fig:method}(a,b)). At query time, early query features and a resource preference determine the route used for continued encoding and document matching (Figure~\ref{fig:method}(c) and Algorithm~\ref{alg:query_retrieval}).
\subsection{Problem formulation}
\label{sec:problem}
Given a text query $q$ and a corpus of document pages $\mathcal X$, visual document retrieval ranks pages by their relevance to the query. We consider a finite set $\R$ of representation configurations, called \emph{routes}, including the native full route $r_F$. Each route specifies encoder depth, embedding width, and query/document token budgets. Under route $r$, matrices $\mathbf Q_r(q)\in\mathbb R^{B_q\times d}$ and $\mathbf D_r(x)\in\mathbb R^{B_x\times d}$ contain query and page vectors as rows. Here $B_q$ and $B_x$ are the retained vector counts and $d$ is their common embedding width; these quantities depend on the route and, for vector counts, the input. Following late interaction \citep{khattab2020colbert,faysse2025colpali}, pages are ranked by
\begin{equation}
s_r(q,x)=\sum_{i=1}^{B_q}\max_{1\le j\le B_x} [\mathbf Q_r(q)]_{i,:}[\mathbf D_r(x)]_{j,:}^{\top}.
\label{eq:maxsim}
\end{equation}
A query-adaptive policy $\pi$ selects $\pi(q)\in\R$ before the remaining query encoding and matching, then ranks the corpus using $s_{\pi(q)}(q,x)$ with query and document representations from that route. A fixed configuration instead uses one route for every query. Document representations are prepared offline; query encoding and matching contribute to online computation. Our objective is to allocate representation budgets to individual queries, balancing retrieval quality and online computation under a specified resource preference.

\begin{figure}[!t]
\centering
\includegraphics[width=\linewidth]{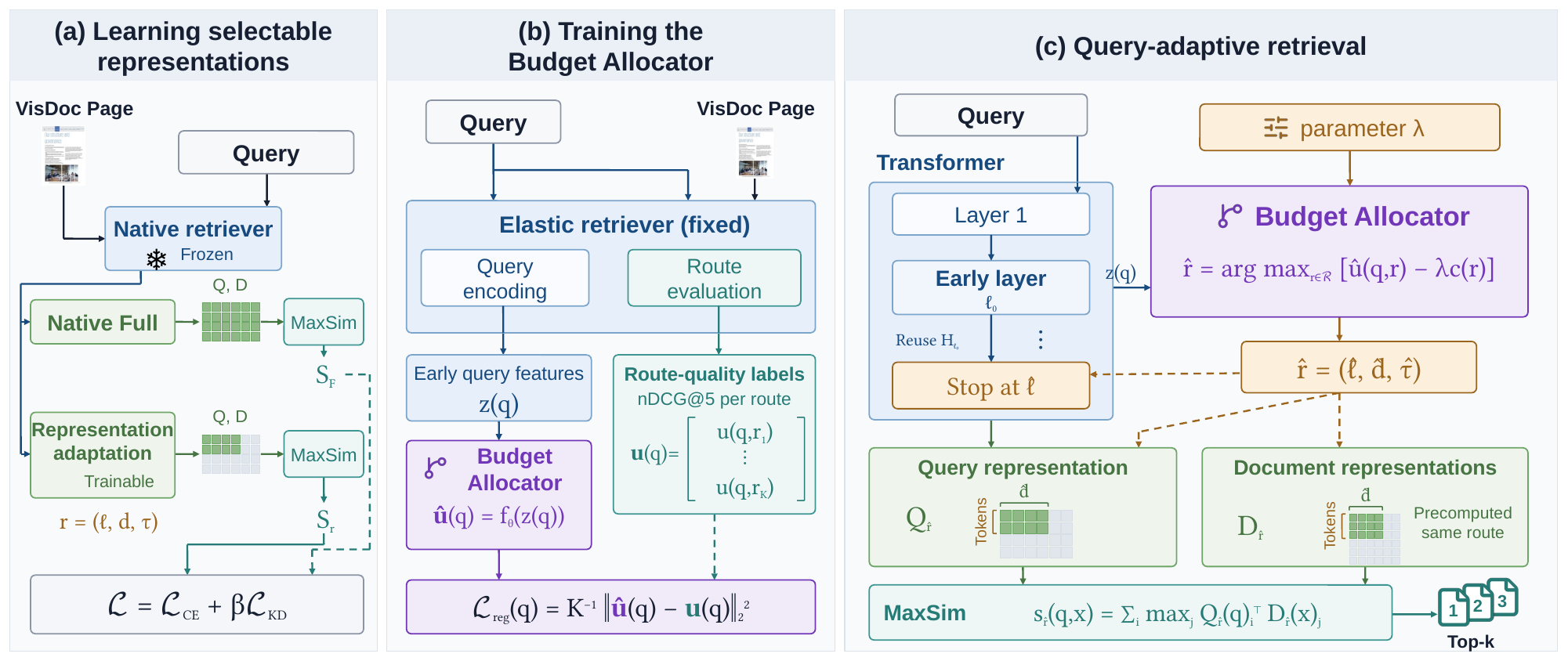}
\caption{\textbf{Training and query-time retrieval in \method.} (a) Retrieval supervision and distillation from frozen Native Full scores $\mathbf S_F$ train compact routes with scores $\mathbf S_r$. (b) Early query features $\mathbf z(q)$ and route utilities $\mathbf u(q)$ train the Budget Allocator. (c) Predicted utilities and resource preference $\lambda$ select a route; encoding stops at its chosen exit, followed by MaxSim matching with document representations from the same route.}
\label{fig:method}
\end{figure}

\begin{algorithm}[t]
\caption{Query-adaptive budget allocation and retrieval in \method}
\label{alg:query_retrieval}
\begin{algorithmic}[1]
\REQUIRE Query $q$; trained elastic retriever and value predictor $f_\theta$; route set $\R$; penalties $c(r)$; resource preference $\lambda$; offline document streams for corpus $\mathcal X$
\ENSURE Document pages ranked for $q$
\algstage{Part 1: Select a query-adaptive budget}
\STATE $\ell_0\gets\min_{r\in\R}\ell(r)$
\STATE Encode $q$ through $\ell_0$ and retain the token-level states $\mathbf H_{\ell_0}(q)$
\STATE Pool $\mathbf H_{\ell_0}(q)$ into $\mathbf z(q)$ and predict $\hat{\mathbf u}(q)=f_\theta(\mathbf z(q))$ (Equation~\ref{eq:values})
\STATE $\hat r=(\hat\ell,\hat d,\hat\tau)\gets\arg\max_{r\in\R}[\hat u(q,r)-\lambda c(r)]$ \algcomment{Joint depth/width/tokens}
\algstage{Part 2: Execute the selected route}
\IF[Reuse early states]{$\hat\ell>\ell_0$}
\STATE Continue encoding from $\mathbf H_{\ell_0}(q)$ through layer $\hat\ell$ to obtain $\mathbf H_{\hat\ell}(q)$
\ENDIF
\IF{$\hat r=r_F$}
\STATE $\mathbf Q_{\hat r}(q)\gets \mathbf E_F(q)$, the native output from the already computed $\mathbf H_L(q)$
\ELSE
\STATE Obtain $\widetilde{\mathbf E}_{\hat\ell}(q)$ from $\mathbf H_{\hat\ell}(q)$ using Equation~\ref{eq:compactbase}
\STATE $\mathbf Q_{\hat r}(q)\gets P_{\hat\tau}\!\left(\operatorname{Norm}_{\mathrm{row}}\!\left(\Pi_{\hat d}\widetilde{\mathbf E}_{\hat\ell}(q)\right)\right)$
\ENDIF
\algstage{Part 3: Match representations from the same route}
\STATE Obtain $\mathbf D_{\hat r}(x)$ for each $x\in\mathcal X$ from the offline streams using Equation~\ref{eq:representation}
\STATE Compute $s_{\hat r}(q,x)$ by MaxSim for each $x\in\mathcal X$ (Equation~\ref{eq:maxsim})
\RETURN Pages sorted by $s_{\hat r}(q,x)$ in descending order
\end{algorithmic}
\end{algorithm}

\subsection{Learning selectable matryoshka representations}
\label{sec:routes}
\label{sec:train}
\paragraph{Representation construction.}
A route $r=(\ell,d,\tau)\in\R$ specifies an encoder exit $\ell\le L$, embedding width $d\le D$, and shared output-token retention ratio $0<\tau\le1$, where $L$ and $D$ are the native depth and width. The finite route set crosses candidate values along all three axes and includes the native full route $r_F=(L,D,1)$. For valid output lengths $n_q$ and $n_x$, it retains $\lceil\tau n_q\rceil$ query vectors and $\lceil\tau n_x\rceil$ document vectors.

For a query or page $a$, let $\mathbf H_\ell(a)$ contain its layer-$\ell$ encoder states and $\mathbf E_F(a)$ its native embeddings for valid output tokens. Token selections follow the native valid-token mask and order. Intermediate heads $G_\ell$ map hidden states to width $D$; final-layer compact routes apply a separate residual map $g$ to the native embeddings:
\begin{equation}
\widetilde{\mathbf E}_\ell(a)=
\begin{cases}
G_\ell\big(\mathbf H_\ell(a)\big), & \ell<L,\\
g\big(\mathbf E_F(a)\big), & \ell=L.
\end{cases}
\label{eq:compactbase}
\end{equation}
Intermediate heads are affine when hidden and retrieval widths differ and low-rank residual maps otherwise. The final-layer map is $g(\mathbf e)=\mathbf e+\mathbf W_{\mathrm{up}}\mathbf W_{\mathrm{down}}\mathbf e$ for each native vector $\mathbf e$. Heads are shared across queries, pages, and all width/token choices at their layer.

Let $\Pi_d$ retain the first $d$ coordinates, $\operatorname{Norm}_{\mathrm{row}}$ apply vector-wise L2 normalization, and $P_\tau$ retain the first $\lceil\tau n_a\rceil$ valid token rows, where $n_a$ is the valid output length. The route output is
\begin{equation}
\mathbf E_r(a)=
\begin{cases}
\mathbf E_F(a), & r=r_F,\\
P_\tau\!\left(\operatorname{Norm}_{\mathrm{row}}\!\left(\Pi_d\widetilde{\mathbf E}_\ell(a)\right)\right), & r\ne r_F.
\end{cases}
\label{eq:representation}
\end{equation}
Thus, compact outputs are renormalized after width truncation, and token selection affects matching vectors only, without pruning encoder tokens. The full route bypasses all adaptation and prefix operations. Setting $\mathbf Q_r(q)=\mathbf E_r(q)$ and $\mathbf D_r(x)=\mathbf E_r(x)$ yields the MaxSim inputs in Equation~\ref{eq:maxsim}.

\paragraph{Training objective.}
We train the adaptation maps using retrieval supervision and distillation from the frozen native full route (Figure~\ref{fig:method}(a)). Supervision across nested widths encourages the leading coordinates to retain retrieval information when later coordinates are omitted \citep{kusupati2022matryoshka}. For a minibatch of $N$ queries and $M$ candidate pages, define score matrices by $[\mathbf S_r]_{i,j}=s_r(q_i,x_j)$ and $[\mathbf S_F]_{i,j}=s_{r_F}(q_i,x_j)$. With $y_i\in\{1,\ldots,M\}$ indexing the paired positive page, we minimize
\begin{equation}
\mathcal L_r=\frac{1}{N}\sum_{i=1}^{N}\left[
\operatorname{CE}([\mathbf S_r]_{i,:},y_i)
+\beta T_{\mathrm{KD}}^2\operatorname{KL}\!\left(
\operatorname{softmax}([\mathbf S_F]_{i,:}/T_{\mathrm{KD}})\,\|\,
\operatorname{softmax}([\mathbf S_r]_{i,:}/T_{\mathrm{KD}})\right)\right].
\label{eq:loss}
\end{equation}
Here $\beta$ weights distillation and $T_{\mathrm{KD}}$ is its temperature. We update only the sampled route's adaptation parameters. After training, we freeze the elastic retriever before generating allocator labels. Appendix~\ref{app:implementation} gives route-sampling and optimization details.

\subsection{Training the Budget Allocator}
\label{sec:router}
The Budget Allocator predicts route utilities from the shallowest candidate exit, $\ell_0=\min_{r\in\R}\ell(r)$, which every query path reaches. At query time, selection uses these early features and route information established offline, without relevance labels, preliminary retrieval results, or full-depth features computed for selection. We pool the token states $\mathbf h_{\ell_0,i}(q)=[\mathbf H_{\ell_0}(q)]_{i,:}^{\top}$ using validity indicators $m_i$ and predict route utilities from the resulting query feature:
\begin{equation}
\mathbf z(q)=\operatorname{norm}_2\!\left(\frac{\sum_i m_i\mathbf h_{\ell_0,i}(q)}{\sum_i m_i}\right),
\qquad \hat{\mathbf u}(q)=f_\theta(\mathbf z(q))\in\mathbb R^{|\R|}.
\label{eq:values}
\end{equation}
Here $\operatorname{norm}_2$ denotes L2 normalization and $\hat u(q,r)$ predicts route $r$'s retrieval utility. The predictor uses a fixed sparse random projection followed by multi-output ExtraTrees regression.

\paragraph{Learning the route utilities.}
For each route, the frozen retriever ranks a deterministic local document pool with the paired page as the positive (Figure~\ref{fig:method}(b)). Query-level nDCG@5 defines $u(q,r)$; the $K=|\R|$ route scores form $\mathbf u(q)=[u(q,r_1),\ldots,u(q,r_K)]$. Supervising the full utility vector retains quality gaps between candidate routes, which a single best-route label would discard. Squared-error regression fits $f_\theta(\mathbf z(q))$ to this vector independently of the resource preference. The fitted predictor remains fixed during calibration and retrieval.

\subsection{Query-adaptive budget allocation and retrieval}
\label{sec:selection}
We pair predicted utility with a fixed structural penalty. Encoding depends on exit depth, while matching scales with query and document token counts and embedding width. Ignoring token rounding, their shared retention ratio gives relative matching volume $\tau^2d/D$, yielding
\begin{equation}
c(r)=\alpha\frac{\ell}{L}+(1-\alpha)\tau^2\frac{d}{D},
\qquad 0\le\alpha\le1,\qquad c(r_F)=1.
\label{eq:cost}
\end{equation}
The fixed $\alpha$ balances the encoding-depth and matching-volume terms; evaluated FLOPs additionally account for input lengths, padding, and backbone-specific operators (Appendix~\ref{app:flops}). For preference $\lambda\ge0$, selection is
\begin{equation}
\hat r=\hat\pi_\lambda(q)=\arg\max_{r\in\R}\big[\hat u(q,r)-\lambda c(r)\big].
\label{eq:decision}
\end{equation}
Larger $\lambda$ favors lower penalties. Because it enters only the selection rule, changing the preference requires neither new utility labels nor predictor refitting.

\paragraph{Theoretical Analysis.}
The selection rule provides control over the structural penalty and relates utility prediction accuracy to the quality of the allocation decision. For fixed predictor, routes, and penalties, $0\le\lambda_1<\lambda_2$ implies $c(\hat\pi_{\lambda_2}(q))\le c(\hat\pi_{\lambda_1}(q))$: increasing $\lambda$ cannot increase the selected structural penalty. This property concerns $c(r)$ rather than evaluated FLOPs. At fixed $\lambda$, define $v(q,r)=u(q,r)-\lambda c(r)$ for a reference utility $u$, and $r^*\in\arg\max_{r\in\R}v(q,r)$. If $\max_r|\hat u(q,r)-u(q,r)|\le\delta_u(q)$, then $v(q,r^*)-v(q,\hat\pi_\lambda(q))\le2\delta_u(q)$. Thus, bounded utility prediction error limits regret in the utility--penalty objective relative to the best available route. Proofs appear in Appendix~\ref{app:proofs}.

\paragraph{Pre-deployment calibration.}
Before deployment, calibration selects a resource preference by minimizing mean structural penalty while limiting quality loss from the native full route. For candidate preferences $\Lambda$ with a nonempty feasible subset, choose
\begin{equation}
\lambda^*\in\arg\min_{\lambda\in\Lambda}\widehat C_{\mathrm{cal}}(\hat\pi_\lambda)
\quad\text{s.t.}\quad
\widehat U_{\mathrm{cal}}(\hat\pi_\lambda)\ge\widehat U_{\mathrm{cal}}(r_F)-\varepsilon_{\mathrm{cal}}.
\label{eq:calibrate}
\end{equation}
Here $\widehat C_{\mathrm{cal}}$ is mean structural penalty, $\widehat U_{\mathrm{cal}}$ is retrieval quality against calibration-corpus judgments, and $\varepsilon_{\mathrm{cal}}$ is the permitted quality decrease.

\paragraph{Executing the selected route.}
Algorithm~\ref{alg:query_retrieval} and Figure~\ref{fig:method}(c) summarize query-time retrieval. After selection, encoding resumes from the retained states $\mathbf H_{\ell_0}(q)$ to $\hat\ell$; if $\hat\ell=\ell_0$, no further layers are needed. Equation~\ref{eq:representation} constructs $\mathbf Q_{\hat r}(q)$ for MaxSim matching against offline document representations $\mathbf D_{\hat r}(x)$ from the same route. The full branch reuses $\mathbf H_L(q)$ without restarting encoding.

Within each compact layer stream, width and token choices reuse prefixes of a maximum-width, maximum-token document representation, with dimension-prefix renormalization. The native full stream remains separate from the transformed final-layer compact stream.

% Main results; placement is left to the normal float mechanism.

\begin{table}[!t]
\centering
\caption{\textbf{Retrieval quality and query budget allocation.} Statistics are macro-averaged over tasks within each benchmark. \method reports mean nDCG@5 ($\times100$) $\pm$ sample SD over three controller seeds. Semicircles show mean query shares: purple for full budget (all three axes maximal), teal for lower budgets. Labels give the lower-budget share. Native baselines always use full budget.}
\label{tab:main}
\begingroup
\fontsize{7.2}{8.7}\selectfont
\setlength{\tabcolsep}{1.2pt}
\renewcommand{\arraystretch}{1.18}
\definecolor{fullcolor}{HTML}{7038AD}
\definecolor{lowercolor}{HTML}{42949E}
\definecolor{rowshade}{HTML}{F3EEF8}
\newcommand{\semipie}[1]{%
\begin{tikzpicture}[x=1mm,y=1mm,baseline=3.5pt]
\path[use as bounding box] (-4,0) rectangle (4,4);
\fill[lowercolor] (0,0)--(4,0) arc[start angle=0,end angle=180,radius=4]--cycle;
\pgfmathsetmacro{\endangle}{180-180*(#1)}
\ifdim #1pt>0pt
\fill[fullcolor] (0,0)--(-4,0) arc[start angle=180,end angle=\endangle,radius=4]--cycle;
\fi
\end{tikzpicture}}
\newcommand{\budgetcell}[2]{\rule[-5pt]{0pt}{15pt}\makebox[46pt][l]{\makebox[23pt][l]{\semipie{#1}}\hspace{1pt}\makebox[22pt][r]{\fontsize{6.2}{7}\selectfont\textcolor{lowercolor}{#2\%}}}}
\newcommand{\meansd}[2]{\makebox[44pt][c]{\makebox[18pt][r]{#1}\hspace{.7pt}{\fontsize{5.8}{6.5}\selectfont$\pm$#2}}}
\newcommand{\meanonly}[1]{\makebox[44pt][c]{\makebox[18pt][r]{#1}\hspace{.7pt}\phantom{\fontsize{5.8}{6.5}\selectfont$\pm$0.00}}}
\begin{tabular*}{\linewidth}{@{\extracolsep{\fill}}>{\raggedright\arraybackslash}p{105pt}*{3}{p{44pt}p{46pt}}@{}}
\toprule
Backbone & \multicolumn{2}{c}{\textbf{ViDoRe V1}} & \multicolumn{2}{c}{\textbf{ViDoRe V2}} & \multicolumn{2}{c}{\textbf{ViDoRe V3 public}} \\
\cmidrule(lr){2-3}\cmidrule(lr){4-5}\cmidrule(lr){6-7}
 & \centering nDCG@5 $\uparrow$ & \centering Budget usage\arraybackslash & \centering nDCG@5 $\uparrow$ & \centering Budget usage\arraybackslash & \centering nDCG@5 $\uparrow$ & \centering Budget usage\arraybackslash \\
\midrule
ColPali  {\fontsize{6.4}{7}\selectfont\citep{faysse2025colpali}} & \meanonly{83.90} & \makebox[46pt][c]{Full only} & \meanonly{53.82} & \makebox[46pt][c]{Full only} & \meanonly{45.74} & \makebox[46pt][c]{Full only} \\
\rowcolor{rowshade}
\hspace{3pt}\textbf{+ \method} & \meansd{83.16}{0.17} & \budgetcell{0.333980200431}{66.60} & \meansd{53.39}{0.23} & \budgetcell{0.163813747489}{83.62} & \meansd{44.28}{0.24} & \budgetcell{0.195615541914}{80.44} \\
\addlinespace[3pt]
ColQwen2  {\fontsize{6.4}{7}\selectfont\citep{faysse2025colpali}} & \meanonly{80.53} & \makebox[46pt][c]{Full only} & \meanonly{43.59} & \makebox[46pt][c]{Full only} & \meanonly{34.64} & \makebox[46pt][c]{Full only} \\
\rowcolor{rowshade}
\hspace{3pt}\textbf{+ \method} & \meansd{80.38}{0.10} & \budgetcell{0.876864893211}{12.31} & \meansd{43.32}{0.31} & \budgetcell{0.802404967125}{19.76} & \meansd{34.53}{0.07} & \budgetcell{0.819313414989}{18.07} \\
\addlinespace[3pt]
ColQwen2.5  {\fontsize{6.4}{7}\selectfont\citep{faysse2025colpali}} & \meanonly{82.27} & \makebox[46pt][c]{Full only} & \meanonly{47.60} & \makebox[46pt][c]{Full only} & \meanonly{42.47} & \makebox[46pt][c]{Full only} \\
\rowcolor{rowshade}
\hspace{3pt}\textbf{+ \method} & \meansd{82.14}{0.18} & \budgetcell{0.787205574077}{21.28} & \meansd{47.50}{0.18} & \budgetcell{0.810074497902}{18.99} & \meansd{42.31}{0.04} & \budgetcell{0.727761323709}{27.22} \\
\addlinespace[3pt]
webAI-ColVec1.1-8B  {\fontsize{6.4}{7}\selectfont\citep{webai2026colvec}} & \meanonly{86.66} & \makebox[46pt][c]{Full only} & \meanonly{60.63} & \makebox[46pt][c]{Full only} & \meanonly{53.69} & \makebox[46pt][c]{Full only} \\
\rowcolor{rowshade}
\hspace{3pt}\textbf{+ \method} & \meansd{86.50}{0.09} & \budgetcell{0.000333333333}{99.97} & \meansd{61.59}{0.57} & \budgetcell{0.000000000000}{100.00} & \meansd{51.69}{0.25} & \budgetcell{0.000000000000}{100.00} \\
\addlinespace[3pt]
Nemotron (8B)  {\fontsize{6.4}{7}\selectfont\citep{moreira2026nemotron}} & \meanonly{89.24} & \makebox[46pt][c]{Full only} & \meanonly{61.19} & \makebox[46pt][c]{Full only} & \meanonly{54.89} & \makebox[46pt][c]{Full only} \\
\rowcolor{rowshade}
\hspace{3pt}\textbf{+ \method} & \meansd{88.99}{0.13} & \budgetcell{0.116793128472}{88.32} & \meansd{61.09}{0.20} & \budgetcell{0.018189847858}{98.18} & \meansd{53.06}{0.25} & \budgetcell{0.032252016293}{96.77} \\
\bottomrule
\end{tabular*}
\par\vspace{5pt}
{\fontsize{7.4}{9}\selectfont
\textcolor{fullcolor}{\rule{7pt}{5pt}}\ Full budget\hspace{15pt}
\textcolor{lowercolor}{\rule{7pt}{5pt}}\ Lower budgets\par}
\endgroup
\end{table}

\section{Experiments}
\label{sec:experiments}
Our experiments examine retrieval quality across backbones and benchmarks, quality--computation trade-offs in three- and two-axis route spaces, and sensitivity to the resource preference.
\subsection{Experimental setup}
\label{sec:setup}
\paragraph{Backbones, benchmarks, and metrics.}
We evaluate ColPali, ColQwen2, ColQwen2.5, webAI-ColVec1.1-8B \citep{webai2026colvec}, and Nemotron ColEmbed V2 (8B) \citep{moreira2026nemotron} on ViDoRe V1, ViDoRe V2, and ViDoRe V3 public \citep{faysse2025colpali,mace2025vidore,loison2026vidorev3}, using each task's complete document corpus. We abbreviate the final two model names as webAI and Nemotron. We report nDCG@5 as the primary metric; supplementary recall@1 and recall@5 results are in Appendix~\ref{app:recall}. \emph{Native Full} denotes the released retrieval path under the evaluation preprocessing. Dataset and scoring details are in Appendix~\ref{app:data}.
\paragraph{Training and evaluation protocol.}
Each backbone supports 64 routes through one elastic checkpoint trained on the ColPali training corpus. The utility predictor is fitted on a controller-training partition grouped by document identity. For the main experiments, one preference $\lambda$ per backbone is selected by Equation~\ref{eq:calibrate} and used across all three benchmarks and controller seeds. Training details and selected preferences are in Appendices~\ref{app:implementation} and~\ref{app:calibration}.
\paragraph{Comparisons and computation.}
The variant study on Industrial compares learned allocation with Fixed 3D Budget, which assigns one depth--width--token configuration to every query, uniform random route selection (random), and three variants with two adaptive axes. All policies share each backbone's trained representation family, allowing us to compare fixed and query-adaptive allocation while holding representation learning constant. Section~\ref{sec:variants} defines the route choices. We report analytical online FLOPs for query encoding, output heads, matching, and Budget Allocator arithmetic. Control definitions and FLOPs calculations are detailed in Appendices~\ref{app:variants} and~\ref{app:flops}.
\subsection{Retrieval quality and representation budgets}
\method preserves retrieval quality close to Native Full while assigning a substantial share of queries to lower budgets (Table~\ref{tab:main}). Across all five backbones and three benchmarks, mean nDCG@5 decreases remain below two percentage points; on V1 and V2, all decreases are below one point. This quality retention accompanies frequent departures from the full-budget route. For example, ColPali assigns 66.60\% of V1 queries to lower budgets while retaining 83.16 nDCG@5 against the native 83.90. On V2, its lower-budget share reaches 83.62\%, with scores of 53.39 versus 53.82. Thus, retaining near-native retrieval quality does not require full representations for every query.

The scope for budget reduction differs markedly across backbones under the quality-preserving preferences. The ColQwen retrievers attain small quality changes with relatively conservative allocations: only 12.31--19.76\% of ColQwen2 queries and 18.99--27.22\% of ColQwen2.5 queries use lower budgets across the three benchmarks. In contrast, webAI assigns 99.97\% of V1 queries to lower budgets while preserving a score of 86.50 against the native 86.66. These operating points indicate different responses to compression: some backbones process nearly all queries at lower budgets with almost no quality loss, whereas others use more conservative budgets when preserving absolute retrieval performance is prioritized.

The two strongest native baselines in our evaluation, webAI and Nemotron, also exhibit the greatest use of lower budgets. For webAI, 99.97\% of V1 queries and 100\% of V2 and V3 public queries use lower budgets. On V2, it improves nDCG@5 from 60.63 to 61.59. No query selects the full-budget route on V2 or V3 public under these settings. Nemotron routes 88.32\%, 98.18\%, and 96.77\% of queries to lower budgets on V1, V2, and V3 public, respectively; its V2 score remains 61.09 against the native 61.19. After adding \method, webAI and Nemotron still occupy the top two positions on all three benchmarks. They therefore retain their absolute retrieval-quality advantage over the other evaluated backbones while assigning the vast majority of queries to lower budgets.
\subsection{Variant and hyperparameter studies}
\label{sec:components}

\subsubsection{Variant study}
\label{sec:variants}
Figure~\ref{fig:variants} evaluates whether learned allocation improves quality--computation trade-offs in both three-axis and two-axis route spaces. Fixed 3D Budget uses the same nominal retention ratio of 25\%, 50\%, 75\%, or 100\% for all three axes and assigns the resulting route to every query. These percentages describe representation budgets rather than FLOPs; exit layers are rounded as in Table~\ref{tab:grids}. Depth fixed+FlexVDR, Width fixed+FlexVDR, and Token Fixed+FlexVDR hold the named axis at 100\% while adapting the other two. Each restricted variant reuses the same elastic checkpoint and utility predictor, applying Equation~\ref{eq:decision} over its 16 remaining routes without refitting. The full three-axis policy selects from all 64 routes.

\begin{figure}[t]
\centering
\includegraphics[width=\linewidth]{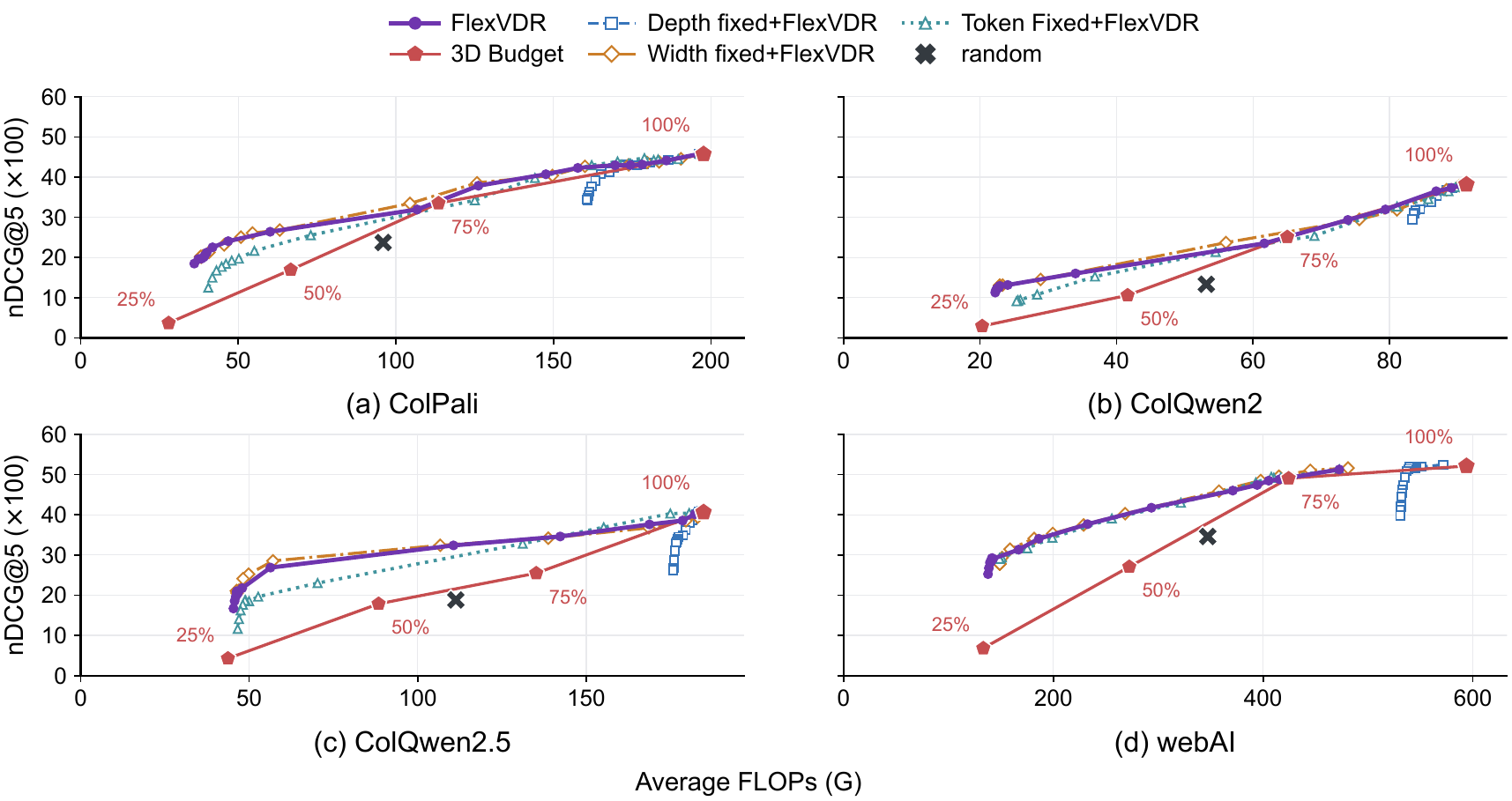}
\caption{\textbf{Variant study on Industrial, ViDoRe V3 public.} nDCG@5 ($\times100$) versus average online GFLOPs per query for 283 queries and the complete 5,244-page corpus. Adaptive curves sweep $\lambda\in[0,1]$; one-axis-fixed variants retain the named axis at 100\%. Red pentagons denote Fixed 3D Budget, using one route for all queries. Labels give the common nominal retention of all three axes; 100\% is Native Full. The random control samples uniformly from all 64 routes.}
\label{fig:variants}
\end{figure}
\paragraph{Learned allocation versus fixed and random choices.}
Query-adaptive allocation can improve retrieval quality while using fewer FLOPs than fixed-budget selection. On ColQwen2.5, \method gains 6.91 percentage points in nDCG@5 over the 75\% Fixed 3D Budget configuration while reducing FLOPs by 18.2\%. At a similar computation level to random selection, it achieves 32.40 nDCG@5 compared with 18.82. All four backbones have learned operating points with both higher quality and lower FLOPs than random selection.
\paragraph{Retaining native retrieval quality.}
\method also supports computation savings while retaining near-native retrieval quality. On webAI, it reduces FLOPs by 20.4\% relative to Native Full with an nDCG@5 decrease of 0.86 percentage points.
\paragraph{Allocation with two adaptive axes.}
Fixing width or token count still allows the remaining axes to recover much of the three-axis policy's quality--computation trade-off (Figure~\ref{fig:variants}). On webAI, the width-fixed and token-fixed variants reduce FLOPs by 19.0\% and 20.4\%, with nDCG@5 decreases of 0.51 and 0.86 percentage points relative to Native Full. These savings are close to those of the three-axis policy while retaining similarly near-native quality. The results suggest that the budget axes offer partially complementary ways to reduce computation: when one axis is unavailable for compression, adapting the others can retain a similar trade-off.

Fixing depth limits the attainable FLOPs reduction more strongly on this workload because every query still executes all encoder layers. Width and output-token reductions save matching work but leave that encoding cost intact (Appendix~\ref{app:flops}). The width-fixed and token-fixed variants retain the option to exit early, allowing substantially lower computation. The relative quality of the two-axis and three-axis policies varies across operating points; their curves do not establish a uniform ranking.
\subsubsection{Hyperparameter analysis}
\label{sec:calibration}
We vary the inference-time preference $\lambda$ with the representations and utility predictor fixed to examine its effect on computation and retrieval quality across five backbones (Figure~\ref{fig:preference}).

\begin{figure}[t]
\centering
\includegraphics[width=\linewidth]{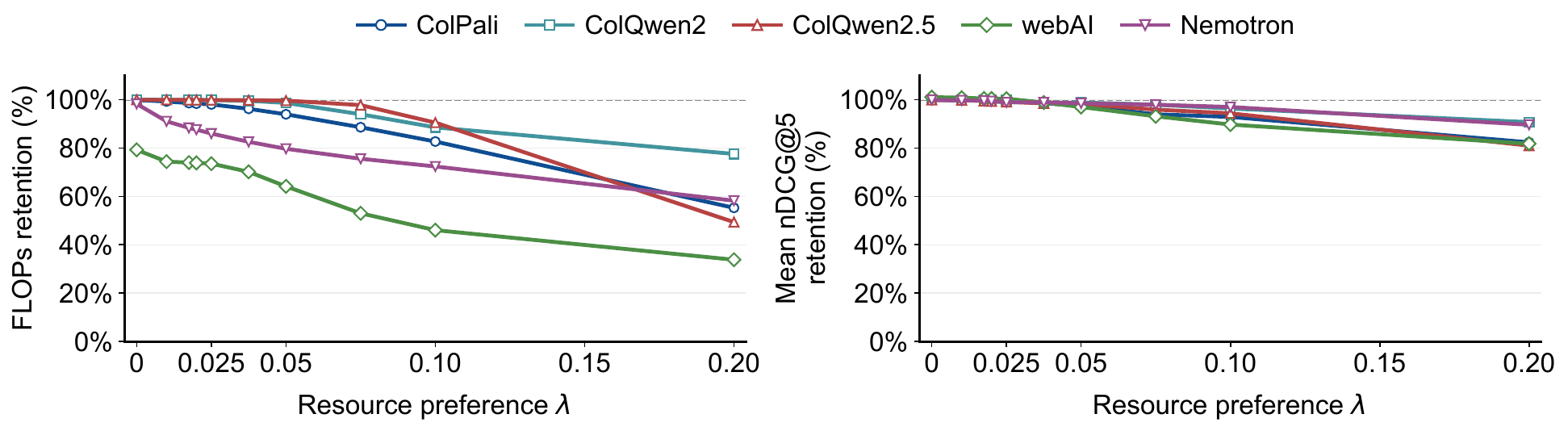}
\par\smallskip
\makebox[0.49\linewidth]{\small (a) Computation on the 1K workload}\hfill
\makebox[0.49\linewidth]{\small (b) Retrieval quality on V1/V2}
\caption{\textbf{Resource-preference sensitivity.} (a) Online FLOPs relative to Native Full on the 1,000-query/1,000-document workload. (b) V1/V2 nDCG@5 retention using all queries and complete corpora: each version's task-macro score is normalized by Native Full, then the two ratios are averaged equally. Dashed lines mark Native Full.}
\label{fig:preference}
\end{figure}
\paragraph{Control over computation.}
Increasing $\lambda$ generally reduces computation by favoring compact routes (Figure~\ref{fig:preference}(a)). At $\lambda=0.1$, webAI uses 46.1\% of Native Full FLOPs on the 1K workload, whereas ColQwen2.5 uses 90.5\%. Computation at the same preference differs across backbones.
\paragraph{Sensitivity of retrieval quality.}
Moderate preferences preserve most native retrieval quality, while stronger preferences incur larger losses (Figure~\ref{fig:preference}(b)). At $\lambda=0.05$, mean V1/V2 nDCG@5 retention ranges from 97.0\% to 98.9\% across the five backbones; at $\lambda=0.2$, it falls to 81.0\%--90.8\%. Appendix~\ref{app:calibration} gives the complete sensitivity grid and aggregation procedure.
\paragraph{Selecting the main preference.}
For the main results, we choose one $\lambda$ per backbone to favor lower computation while keeping retrieval quality close to Native Full. Equation~\ref{eq:calibrate} implements this principle by minimizing mean structural penalty among candidates whose calibration nDCG@5 decreases by at most 0.005 (0.5 percentage points). The selected values are 0.025 for ColPali and ColQwen2, 0.0175 for ColQwen2.5 and Nemotron, and 0.020 for webAI. Each value is then fixed across V1, V2, V3 public, and all controller seeds. Calibration uses the structural penalty as a computation proxy; FLOPs savings are evaluated separately. Appendix~\ref{app:calibration} details the calibration subset and procedure.
\section{Conclusion}
We introduced \method to allocate representation budgets to individual queries in visual document retrieval. Its Budget Allocator uses early query features to choose among a joint family of depth--width--token representations, including the native full route. Separating utility prediction from resource preference lets the policy support different operating points without retraining. Across five backbones on ViDoRe V1, V2, and V3 public, the learned policies retain much of the native retrieval quality while selecting compact routes, with different quality and allocation responses across backbones. By learning the retrieval value of available representations, \method makes their allocation a query-level decision that can adapt to the desired resource preference.

\section*{AI use statement}
We used GPT-5.3-Codex to assist with code debugging and GPT-5.5 for grammar checking and language polishing. These tools also assisted with feedback on experimental methodology, interpretation of results, and refinement of mathematical statements and proofs. The authors have reviewed and verified all AI-assisted content. The authors take responsibility for the final content of the paper, including all AI-assisted text, code, claims, and artifacts.

\section*{Reproducibility statement}
Appendix~\ref{app:implementation} specifies the released model identifiers, route construction, data partitions, and training configurations for the representations and Budget Allocator. Appendix~\ref{app:data} describes the evaluation datasets, preprocessing, scoring, and aggregation procedures. Appendix~\ref{app:calibration} details preference calibration and sensitivity evaluation, while Appendix~\ref{app:flops} provides the formulas and counting conventions for the FLOPs estimates. Appendix~\ref{app:proofs} states the assumptions and gives proofs of the decision properties.

\section*{Ethics statement}
This work studies visual document retrieval using released models and existing document benchmarks. Retrieval errors and biases inherited from these models and corpora may affect which documents are surfaced to users. Applications should respect document access permissions, privacy requirements, and model and dataset licenses, while preserving attribution to retrieved sources.
\bibliographystyle{iclr2027_conference}
\bibliography{references}
\clearpage
\appendix
\section{Implementation and route construction}
\label{app:implementation}
\paragraph{Released model identities.}
The evaluated models are \texttt{vidore/colpali-v1.3-merged}, \texttt{vidore/colqwen2-v1.0-merged}, \texttt{vidore/colqwen2.5-v0.2} with its matching \texttt{vidore/colqwen2.5-base}, \texttt{webAI-Official/webAI-ColVec1.1-8b}, and \texttt{nvidia/nemotron-colembed-vl-8b-v2}. The ColQwen2.5 release adapter remains part of the frozen native model.

\begin{table}[htbp]
\centering\small
\caption{Route grid. Each depth--width pair is crossed with token retention ratios $\{0.25,0.50,0.75,1.00\}$. Native Full uses maximal depth and width with token retention 1.}
\label{tab:grids}
\begin{tabular}{lll}
\toprule
Backbone & Layers & Dimensions \\
\midrule
ColPali & 4, 9, 14, 18 & 32, 64, 96, 128 \\
ColQwen2 & 7, 14, 21, 28 & 32, 64, 96, 128 \\
ColQwen2.5 & 9, 18, 27, 36 & 32, 64, 96, 128 \\
webAI & 8, 16, 24, 32 & 160, 320, 480, 640 \\
Nemotron & 9, 18, 27, 36 & 1024, 2048, 3072, 4096 \\
\bottomrule
\end{tabular}
\end{table}

\paragraph{Candidate grid and structural penalty.}
Each axis uses nominal retention ratios $\{1/4,1/2,3/4,1\}$, giving 64 routes with integer layer coordinates listed in Table~\ref{tab:grids}. Equation~\ref{eq:cost} uses $\alpha=0.7$ for all backbones and resource preferences.

\paragraph{Valid tokens and rounding.}
The mask selects the first $\lceil\tau n\rceil$ valid tokens according to the released model's output order, independently for each sequence. It does not select the first positions of a padded tensor. The retained page representation can include valid non-patch tokens admitted by the native mask, so ``token prefix'' should not be read as an importance-ranked patch set. ColPali's selected layers are $4,9,14,18$, whose actual retained proportions are $22.22\%,50\%,77.78\%,100\%$. Cost and the reported depth-retention statistics use actual layers, not the nominal quartile labels.

\paragraph{Projection initialization.}
Low-rank adaptation maps use rank 32, normally initialized down-projections with standard deviation 0.02, and zero-initialized up-projections. Linear intermediate heads copy the released projection when available; these compact heads remain separate from the native retrieval path.

\paragraph{Optimization.}
The retriever is frozen; the projection bank is trained with AdamW, learning rate $5\times10^{-5}$, weight decay 0.01, 3\% linear warmup, and cosine decay. Each run uses three epochs over 118,195 ColPali training query--page rows, batch size four on each of eight GPUs (global batch size 32), and cross-device document negatives. The resulting checkpoint is at step 11,082. Elastic training and the deterministic shuffled route cycle use seed 42. Equation~\ref{eq:loss} uses $\beta=1$ and $T_{\mathrm{KD}}=0.05$. The CE term uses unscaled MaxSim logits; the temperature and its squared factor apply to KL distillation.

\paragraph{Budget Allocator fitting.}
Document identifiers are assigned to an 80/10/10 split using a seeded SHA-256 hash (seed 42), yielding 95,095 controller-training queries, 11,566 validation queries, and 11,534 test queries. Route utilities are generated for all 118,195 training queries. The feature transform and regressor are fitted only on the controller-training partition; the elastic retriever was trained on all 118,195 rows. Candidate groups contain 64 unique documents from the same source and split where available, with deterministic same-split filling for incomplete groups. Each training query uses its paired page as the labeled positive in this local pool. All queries from a page retain the same split; merged cache validation checks disjoint hashes of image filenames, not pixel-content hashes. These local single-positive utilities differ from external multi-positive evaluation.

The early pooled feature is rounded consistently with the FP16 cache and L2-normalized. SparseRandomProjection has 256 output features and dense output. ExtraTreesRegressor uses 100 estimators, minimum leaf size 20, and a feature fraction of 0.5. Projection and regressor use the corresponding controller seed from $\{42,123,2026\}$. These three separately fitted controllers share one elastic checkpoint per backbone; averaging their reported metrics is not prediction ensembling.

\paragraph{Materializing document representation streams.}
Compact width/token routes share prefixes within a layer stream, with dimension-prefix renormalization at use time. Different layers require separate streams; the transformed final-layer compact stream must also be stored separately from Native Full or computed on demand. For materialized streams $\mathcal S$ and $b$ bytes per scalar, the dense vector payload is
\[
S_{\mathrm{payload}}=b\sum_{s\in\mathcal S}\sum_{x\in\mathcal X} B^{\max}_{x,s}D^{\max}_s.
\]
Here $B^{\max}_{x,s}$ is the maximum valid vector count stored for page $x$ in stream $s$, and $D^{\max}_s$ is that stream's maximum embedding width. Metadata, indexing structures, alignment, quantization, and temporary buffers add to this payload.

\section{External datasets and scoring}
\label{app:data}
ViDoRe V1 comprises ten task configurations with 3,871 queries, and ViDoRe V2 comprises four with 1,152 queries. Table~\ref{tab:sizes} lists the query and corpus sizes for each task.

\begin{table}[htbp]
\centering\small
\caption{ViDoRe V1/V2 query and corpus sizes. Each task includes all judged queries and unique corpus pages. ESG v2 and ESG Human v2 share a corpus but use different queries.}
\label{tab:sizes}
\begin{tabular}{lrr}
\toprule
Dataset & Queries & Pages \\
\midrule
DocVQA & 451 & 500 \\
InfoVQA & 494 & 500 \\
TabFQuAD & 280 & 70 \\
ArXivQA & 500 & 500 \\
TAT-DQA & 1646 & 277 \\
Shift & 100 & 1000 \\
Synthetic AI & 100 & 968 \\
Synthetic Energy & 100 & 977 \\
Synthetic Government & 100 & 972 \\
Synthetic Healthcare & 100 & 965 \\
ESG v2 & 228 & 1538 \\
Biomedical v2 & 640 & 1016 \\
Economics v2 & 232 & 452 \\
ESG Human v2 & 52 & 1538 \\
\bottomrule
\end{tabular}
\end{table}

\paragraph{Controlled inputs and precision.}
For V1/V2, every configuration of a backbone uses the same input queries, page images, relevance judgments, and scoring precision. Images use a 262,144-pixel cap and a 256-visual-token processor limit where applicable; ColPali uses its native fixed-resolution processor. Retrieval embeddings are stored in FP16, and exact MaxSim scoring is computed in FP32. The native full reference is evaluated under this same preprocessing and numerical representation.

\paragraph{Quality replay.}
Full-forward hidden states provide query--route rankings against each task's complete corpus. V1/V2 comparisons replay cached utilities and controller predictions without refitting or new GPU forward passes. This offline enumeration is used only for evaluation; query-time selection uses early features and the fixed penalty table.

\paragraph{Identity and relevance.}
For V1, corpus identity is the image filename. Repeated query text is collapsed and its positive-page set is the union of matching examples. Corpus-only rows remain candidate pages. The first four dataset identifiers include the benchmark's \texttt{test\_subsampled} suffix; ``full corpus'' refers to all pages of those benchmark configurations, not the source datasets before benchmark construction. For V2, corpus, query, and qrels tables are joined by their explicit identifiers. Positive integer grades are preserved, and repeated judgments use the maximum grade. Only judged queries are scored. Source reports list any exclusions explicitly.

\paragraph{Metrics.}
We use \texttt{pytrec\_eval} with \texttt{ndcg\_cut.5}, \texttt{recall.1,5}, and \texttt{recip\_rank}, including its relevance-gain and tie conventions. Recall is the fraction of relevant pages retrieved; for queries with multiple positives, recall@1 is not simply hit rate. For each controller seed in $\{42,123,2026\}$, we first average tasks equally within each V1/V2 version, then compute the mean and sample SD (ddof$=1$) of the three seed-level means. Preferences and the elastic checkpoint are fixed across those seeds. V1/V2 supplementary overall results and component controls use the unweighted 14-task mean, so task size does not determine weight. Fixed-route and random-assignment controls use seed 42. The exported numerical companion also includes reciprocal rank and all fixed-route metrics.

\begin{table}[t]
\centering\scriptsize
\caption{Per-task nDCG@5 for Native Full / Router on V1/V2. Router scores average three controller seeds at the main preferences.}
\label{tab:perdataset}
\begin{tabular}{lrrrrr}
\toprule
Dataset & ColPali & ColQwen2 & ColQwen2.5 & webAI & Nemotron \\
\midrule
DocVQA & 0.575 / 0.571 & 0.555 / 0.556 & 0.558 / 0.560 & 0.585 / 0.599 & 0.626 / 0.627 \\
InfoVQA & 0.836 / 0.834 & 0.865 / 0.865 & 0.842 / 0.840 & 0.896 / 0.891 & 0.905 / 0.900 \\
TabFQuAD & 0.859 / 0.861 & 0.891 / 0.891 & 0.885 / 0.884 & 0.967 / 0.968 & 0.975 / 0.971 \\
ArXivQA & 0.825 / 0.821 & 0.863 / 0.862 & 0.866 / 0.862 & 0.899 / 0.899 & 0.915 / 0.912 \\
TAT-DQA & 0.696 / 0.694 & 0.676 / 0.676 & 0.625 / 0.626 & 0.669 / 0.664 & 0.731 / 0.729 \\
Shift & 0.748 / 0.736 & 0.576 / 0.570 & 0.651 / 0.644 & 0.794 / 0.790 & 0.870 / 0.869 \\
Synthetic AI & 0.974 / 0.952 & 0.952 / 0.950 & 0.970 / 0.965 & 0.989 / 0.993 & 0.980 / 0.984 \\
Synthetic Energy & 0.950 / 0.938 & 0.888 / 0.889 & 0.956 / 0.951 & 0.961 / 0.958 & 0.957 / 0.950 \\
Synthetic Government & 0.952 / 0.943 & 0.917 / 0.910 & 0.939 / 0.944 & 0.938 / 0.935 & 0.984 / 0.982 \\
Synthetic Healthcare & 0.975 / 0.966 & 0.871 / 0.869 & 0.935 / 0.938 & 0.968 / 0.953 & 0.981 / 0.976 \\
ESG v2 & 0.525 / 0.521 & 0.414 / 0.413 & 0.389 / 0.390 & 0.562 / 0.566 & 0.561 / 0.569 \\
Biomedical v2 & 0.557 / 0.547 & 0.549 / 0.547 & 0.579 / 0.576 & 0.656 / 0.654 & 0.653 / 0.639 \\
Economics v2 & 0.484 / 0.489 & 0.410 / 0.406 & 0.499 / 0.497 & 0.553 / 0.577 & 0.584 / 0.576 \\
ESG Human v2 & 0.586 / 0.580 & 0.370 / 0.367 & 0.437 / 0.436 & 0.655 / 0.667 & 0.649 / 0.659 \\
\bottomrule
\end{tabular}
\end{table}

\subsection{ViDoRe V3 public evaluation}
\label{app:v3}
ViDoRe V3 public evaluation uses 2,419 English queries from eight public datasets and their complete document corpora (Table~\ref{tab:v3sizes}).

\begin{table}[htbp]
\centering\small
\caption{ViDoRe V3 public query and corpus sizes. All corpora are complete; dataset revisions are recorded in the numerical companion.}
\label{tab:v3sizes}
\begin{tabular}{lrr}
\toprule
Dataset & English queries & Corpus pages \\
\midrule
hr & 318 & 1,110 \\
finance\_en & 309 & 2,942 \\
industrial & 283 & 5,244 \\
pharmaceuticals & 364 & 2,313 \\
computer\_science & 215 & 1,360 \\
energy & 308 & 2,225 \\
physics & 302 & 1,674 \\
finance\_fr & 320 & 2,384 \\
\midrule
Total & 2,419 & 19,252 \\
\bottomrule
\end{tabular}
\end{table}

V3 retains relevance grades 1 and 2 and uses \texttt{pytrec\_eval} for nDCG@5/10, recall@1/5/10, and reciprocal rank. For each of the three controller seeds, retrieval metrics and full-route shares are averaged equally over the eight datasets; we report the mean and sample SD of these seed-level averages. The native reference, elastic checkpoint, and per-backbone preference remain fixed. Table~\ref{tab:v3metrics} reports supplementary nDCG@10 and MRR.

V3 uses the preprocessing above. Native Full retains its original embedding dtype, whereas routed outputs, including the full route, are stored in FP16 before FP32 MaxSim scoring. Quality is evaluated by full-forward encoding and route replay.

\begin{table}[htbp]
\centering\small
\caption{Supplementary ViDoRe V3 public metrics at the main preferences. Router reports mean $\pm$ sample SD over three controller seeds after equal averaging across eight datasets. Native Full is fixed.}
\label{tab:v3metrics}
\begin{tabular}{lrrr}
\toprule
Backbone & Full nDCG@10 & Router nDCG@10 & Router MRR \\
\midrule
ColPali & 0.4821 & $0.4676\pm0.0019$ & $0.5890\pm0.0022$ \\
ColQwen2 & 0.3714 & $0.3692\pm0.0003$ & $0.4672\pm0.0008$ \\
ColQwen2.5 & 0.4532 & $0.4509\pm0.0011$ & $0.5617\pm0.0019$ \\
webAI & 0.5634 & $0.5452\pm0.0015$ & $0.6538\pm0.0025$ \\
Nemotron & 0.5731 & $0.5551\pm0.0026$ & $0.6695\pm0.0037$ \\
\bottomrule
\end{tabular}
\end{table}

\Needspace{24\baselineskip}
\subsection{Supplementary recall results}
\label{app:recall}
\begin{table}[H]
\centering\small
\caption{\textbf{Recall across ViDoRe benchmarks.} R@1 and R@5 ($\times100$) are macro-averaged over tasks within each benchmark. Shaded \method rows average three controller seeds at the preferences in Table~\ref{tab:main}; native baselines are fixed.}
\label{tab:recall}
\begingroup
\setlength{\tabcolsep}{2pt}
\renewcommand{\arraystretch}{1.08}
\begin{tabular*}{\linewidth}{@{\extracolsep{\fill}}>{\raggedright\arraybackslash}p{0.37\linewidth}rrrrrr@{}}
\toprule
\rowcolor{datasetgray}Backbone & \multicolumn{2}{c}{\textbf{ViDoRe V1}} & \multicolumn{2}{c}{\textbf{ViDoRe V2}} & \multicolumn{2}{c}{\textbf{ViDoRe V3 public}} \\
 & R@1 $\uparrow$ & R@5 $\uparrow$ & R@1 $\uparrow$ & R@5 $\uparrow$ & R@1 $\uparrow$ & R@5 $\uparrow$ \\
\midrule
ColPali & 76.58 & 89.87 & 24.92 & 50.17 & 21.36 & 42.08 \\
\rowcolor{routepurple!8}
ColPali \textbf{+ \method} & 75.76 & 89.21 & 25.69 & 49.18 & 20.65 & 41.07 \\
\addlinespace[2pt]
ColQwen2 & 72.14 & 87.21 & 19.70 & 40.49 & 15.37 & 33.35 \\
\rowcolor{routepurple!8}
ColQwen2 \textbf{+ \method} & 71.94 & 87.11 & 19.53 & 40.42 & 15.17 & 33.37 \\
\addlinespace[2pt]
ColQwen2.5 & 74.86 & 88.46 & 21.13 & 44.64 & 18.90 & 39.61 \\
\rowcolor{routepurple!8}
ColQwen2.5 \textbf{+ \method} & 74.70 & 88.37 & 21.22 & 44.44 & 18.83 & 39.51 \\
\addlinespace[2pt]
webAI-ColVec1.1-8B & 80.05 & 91.67 & 28.31 & 56.76 & 24.63 & 50.08 \\
\rowcolor{routepurple!8}
webAI-ColVec1.1-8B \textbf{+ \method} & 79.81 & 91.77 & 28.62 & 58.00 & 23.50 & 48.63 \\
\addlinespace[2pt]
Nemotron ColEmbed V2 (8B) & 83.96 & 93.43 & 30.60 & 56.27 & 25.68 & 50.76 \\
\rowcolor{routepurple!8}
Nemotron ColEmbed V2 (8B)\newline \textbf{+ \method} & 83.43 & 93.29 & 30.60 & 56.16 & 24.75 & 49.16 \\
\bottomrule
\end{tabular*}
\endgroup
\end{table}

\section{Variant definitions and supplementary representation controls}
\label{app:training}
\subsection{Variant study on Industrial}
\label{app:variants}
The Industrial study evaluates ColPali, ColQwen2, ColQwen2.5, and webAI on 283 queries and all 5,244 pages, with one trained elastic checkpoint and the seed-42 utility predictor per backbone. Preprocessing and precision follow Appendix~\ref{app:v3}; FLOPs use the actual token lengths of these queries and this corpus. Variant definitions are given in Section~\ref{sec:variants}. Fixed 3D Budget uses our trained representation family and does not reproduce MM-Matryoshka training.

For random selection, \texttt{default\_rng(42)} draws one uniform route index from $\{0,\ldots,63\}$ per query in the fixed query order. The normalized indices are shared across backbones and fixed across preferences.

All four adaptive policies use
\[
\Lambda_{\mathrm{variant}}=\{0,.01,.0175,.02,.025,.0375,.05,.075,.1,.2,.3,.4,.5,.6,.7,.8,.9,1\}.
\]
Quality scoring uses batches of four queries and eight pages; online FLOPs use one query and the same eight-page matching blocks (Appendix~\ref{app:flops}). The ColQwen2.5 comparison in Section~\ref{sec:variants} uses $\lambda=0.1$. The webAI near-native example uses $\lambda=0$; its depth-, width-, and token-fixed comparisons use $\lambda=0.05$, 0, and 0, respectively.

\subsection{Supplementary V1/V2 controls}
Table~\ref{tab:controls} reports seed-42 controls on the 14-task V1/V2 suite; Router denotes the Budget Allocator. The listed fixed routes reuse the trained all-axis checkpoint and assign the same configuration to every query, unlike the Industrial variants that retain two adaptive axes.

\begin{table}[htbp]
\centering\small
\caption{\textbf{Routing and fixed-budget comparisons on V1/V2.} nDCG@5 is macro-averaged over 14 tasks. Policies share the trained representation family; the Budget Allocator and uniform random control use seed 42. L/D/T denotes nominal depth/width/token retention for fixed routes.}
\label{tab:controls}
\begingroup
\setlength{\tabcolsep}{4pt}
\renewcommand{\arraystretch}{1.12}
\begin{tabular*}{\linewidth}{@{\extracolsep{\fill}}llrrr@{}}
\toprule
Policy / configuration & L/D/T (\%) & ColPali & ColQwen2 & ColQwen2.5 \\
\midrule
\rowcolor{datasetgray}\multicolumn{5}{l}{\textbf{Reference and routing policies}} \\
Native Full & 100/100/100 & 0.7531 & 0.6998 & 0.7237 \\
random & All 64 routes & 0.4952 & 0.4112 & 0.4463 \\
\rowcolor{routepurple!8}
\textbf{\method} & Adaptive & 0.7485 & 0.6978 & 0.7205 \\
\addlinespace[4pt]
\rowcolor{datasetgray}\multicolumn{5}{l}{\textbf{Fixed configurations}} \\
Width reduction & 100/75/100 & 0.7452 & 0.6864 & 0.7099 \\
Token reduction & 100/100/75 & 0.7435 & 0.6774 & 0.6946 \\
Depth reduction & 75/100/100 & 0.6743 & 0.6246 & 0.6273 \\
Joint 75\% & 75/75/75 & 0.6701 & 0.6263 & 0.6055 \\
Joint 50\% & 50/50/50 & 0.4853 & 0.4028 & 0.4850 \\
\bottomrule
\end{tabular*}
\endgroup
\end{table}

\begin{table}[htbp]
\centering\small
\caption{Final-layer nDCG@5 for native prefix truncation and trained compact routes. Subscripts denote the shared width/token retention percentage; depth remains full.}
\label{tab:prefix}
\begin{tabular}{lrrrr}
\toprule
Backbone & Native$_{50}$ & Trained$_{50}$ & Native$_{75}$ & Trained$_{75}$ \\
\midrule
ColPali & 0.6732 & 0.6738 & 0.7359 & 0.7344 \\
ColQwen2 & 0.5310 & 0.5283 & 0.6676 & 0.6675 \\
ColQwen2.5 & 0.5424 & 0.5411 & 0.6824 & 0.6810 \\
webAI & 0.7377 & 0.7400 & 0.7808 & 0.7815 \\
Nemotron & 0.7356 & 0.7374 & 0.8059 & 0.8006 \\
\bottomrule
\end{tabular}
\end{table}

\section{Confirmation scores and full-query evaluation}
\label{app:confirmation}
The 224 confirmation queries comprise 16 per V1/V2 task, disjoint from the calibration queries but included in the full evaluation. All three controllers use the same confirmation queries and fixed per-backbone preference. For Table~\ref{tab:confirmationdetail}, we average each query's Router $-$ Full difference over seeds, bootstrap queries within each task for 1,000 draws, and average tasks equally. All five 95\% intervals include zero.

\begin{table}[htbp]
\centering\small
\caption{Confirmation-set nDCG@5 at the main preferences. Router SD is the sample SD over three controller seeds. The 95\% query-bootstrap intervals are for Router minus Native Full in percentage points. Mean penalty averages the selected $c(r)$.}
\label{tab:confirmationdetail}
\begin{tabular}{lrrrrr}
\toprule
Backbone & Full & Router & Router SD & 95\% CI of $\Delta$ (pp) & Mean penalty \\
\midrule
ColPali & 0.7173 & 0.7140 & 0.0063 & [-1.52, +0.79] & 0.9039 \\
ColQwen2 & 0.6765 & 0.6778 & 0.0018 & [-0.18, +0.53] & 0.9879 \\
ColQwen2.5 & 0.7203 & 0.7199 & 0.0037 & [-0.40, +0.29] & 0.9838 \\
webAI & 0.7657 & 0.7614 & 0.0027 & [-2.32, +1.66] & 0.6297 \\
Nemotron & 0.7820 & 0.7854 & 0.0006 & [-0.83, +1.49] & 0.7862 \\
\bottomrule
\end{tabular}
\end{table}

\paragraph{Full-query evaluation of the same selected preferences.}
\label{app:fullsuite}
Table~\ref{tab:fullsuite} aggregates all 5,023 V1/V2 queries, including calibration and confirmation subsets, at the same preferences and controller seeds. Table~\ref{tab:main} separates these records by benchmark. Overall metric means equal $(10\mu_{\mathrm{V1}}+4\mu_{\mathrm{V2}})/14$; overall SDs are computed from seed-level overall means, not averaged from benchmark SDs. The nDCG@5 changes from Native Full are $-0.655\pm0.173$, $-0.182\pm0.108$, $-0.121\pm0.176$, $+0.159\pm0.216$, and $-0.207\pm0.132$ percentage points for ColPali, ColQwen2, ColQwen2.5, webAI, and Nemotron, respectively (mean $\pm$ sample SD).

\begin{table}[htbp]
\centering\footnotesize
\setlength{\tabcolsep}{2.3pt}
\caption{V1/V2 results on all 5,023 queries, including calibration and confirmation subsets. Metrics are macro-averaged over 14 tasks; Router reports mean $\pm$ sample SD over three controller seeds. Full (\%) is the task-macro query share selecting Native Full.}
\label{tab:fullsuite}
\begin{tabular}{@{}llrrrrr@{}}
\toprule
Backbone & Policy & nDCG@5 $\uparrow$ & R@1 $\uparrow$ & R@5 $\uparrow$ & MRR $\uparrow$ & Full (\%) \\
\midrule
ColPali & Full & 0.7531 & 0.6182 & 0.7853 & 0.7757 & 100.00 \\
 & Router & $0.7465\pm0.0017$ & $0.6146\pm0.0021$ & $0.7777\pm0.0016$ & $0.7725\pm0.0026$ & $28.54\pm2.52$ \\
ColQwen2 & Full & 0.6998 & 0.5715 & 0.7386 & 0.7227 & 100.00 \\
 & Router & $0.6979\pm0.0011$ & $0.5696\pm0.0007$ & $0.7377\pm0.0007$ & $0.7205\pm0.0013$ & $85.56\pm1.45$ \\
ColQwen2.5 & Full & 0.7237 & 0.5951 & 0.7594 & 0.7497 & 100.00 \\
 & Router & $0.7225\pm0.0018$ & $0.5942\pm0.0017$ & $0.7582\pm0.0028$ & $0.7487\pm0.0014$ & $79.37\pm2.90$ \\
webAI & Full & 0.7922 & 0.6527 & 0.8169 & 0.8158 & 100.00 \\
 & Router & $0.7938\pm0.0022$ & $0.6518\pm0.0011$ & $0.8212\pm0.0020$ & $0.8180\pm0.0025$ & $0.02\pm0.01$ \\
Nemotron & Full & 0.8123 & 0.6872 & 0.8281 & 0.8373 & 100.00 \\
 & Router & $0.8102\pm0.0013$ & $0.6833\pm0.0024$ & $0.8268\pm0.0008$ & $0.8364\pm0.0019$ & $8.86\pm1.16$ \\
\bottomrule
\end{tabular}
\end{table}

\section{Routing distributions at the selected preferences}
\label{app:distributions}
\begin{table}[htbp]
\centering\small
\caption{Full-route query shares and mean retained budgets (\%). Values are task-macro means $\pm$ sample SD over three controller seeds within each benchmark; the overall V1/V2 group averages all 14 tasks equally. Depth and width are relative to native values; Tokens denotes $\tau$ before rounding. Full counts queries selecting $(L,D,1)$; SDs are in percentage points.}
\label{tab:budgetversions}
\begin{tabular}{lrrrr}
\toprule
Backbone & Full (\%) & Depth (\%) & Width (\%) & Tokens (\%) \\
\midrule
\rowcolor{datasetgray}\multicolumn{5}{l}{\textbf{ViDoRe V1}} \\
ColPali & $33.40\pm2.08$ & $99.66\pm0.14$ & $82.83\pm0.16$ & $91.69\pm0.51$ \\
ColQwen2 & $87.69\pm1.73$ & $100.00\pm0.00$ & $97.14\pm0.40$ & $99.73\pm0.09$ \\
ColQwen2.5 & $78.72\pm2.36$ & $100.00\pm0.00$ & $94.77\pm0.68$ & $99.86\pm0.06$ \\
webAI & $0.03\pm0.01$ & $74.97\pm0.11$ & $36.42\pm0.54$ & $97.80\pm0.37$ \\
Nemotron & $11.68\pm1.42$ & $98.39\pm0.34$ & $33.83\pm1.03$ & $98.41\pm0.10$ \\
\rowcolor{datasetgray}\multicolumn{5}{l}{\textbf{ViDoRe V2}} \\
ColPali & $16.38\pm3.74$ & $99.84\pm0.03$ & $76.08\pm1.02$ & $88.84\pm1.03$ \\
ColQwen2 & $80.24\pm0.73$ & $100.00\pm0.00$ & $95.57\pm0.13$ & $99.42\pm0.15$ \\
ColQwen2.5 & $81.01\pm4.41$ & $100.00\pm0.00$ & $95.40\pm0.96$ & $99.82\pm0.20$ \\
webAI & $0.00\pm0.00$ & $74.99\pm0.01$ & $30.75\pm0.76$ & $97.98\pm1.00$ \\
Nemotron & $1.82\pm0.65$ & $98.71\pm0.15$ & $26.41\pm0.50$ & $98.38\pm0.25$ \\
\rowcolor{datasetgray}\multicolumn{5}{l}{\textbf{ViDoRe V3 public}} \\
ColPali & $19.56\pm1.55$ & $99.17\pm0.18$ & $76.86\pm0.44$ & $87.99\pm0.57$ \\
ColQwen2 & $81.93\pm2.31$ & $100.00\pm0.00$ & $96.01\pm0.43$ & $99.42\pm0.15$ \\
ColQwen2.5 & $72.78\pm2.43$ & $100.00\pm0.00$ & $93.89\pm0.44$ & $99.21\pm0.22$ \\
webAI & $0.00\pm0.00$ & $74.14\pm0.56$ & $31.97\pm0.70$ & $96.53\pm0.39$ \\
Nemotron & $3.23\pm0.98$ & $94.95\pm0.91$ & $27.54\pm0.70$ & $97.40\pm0.08$ \\
\rowcolor{datasetgray}\multicolumn{5}{l}{\textbf{V1/V2 overall (14 tasks)}} \\
ColPali & $28.54\pm2.52$ & $99.71\pm0.10$ & $80.90\pm0.25$ & $90.88\pm0.61$ \\
ColQwen2 & $85.56\pm1.45$ & $100.00\pm0.00$ & $96.69\pm0.30$ & $99.64\pm0.11$ \\
ColQwen2.5 & $79.37\pm2.90$ & $100.00\pm0.00$ & $94.95\pm0.75$ & $99.85\pm0.03$ \\
webAI & $0.02\pm0.01$ & $74.98\pm0.08$ & $34.80\pm0.31$ & $97.85\pm0.54$ \\
Nemotron & $8.86\pm1.16$ & $98.48\pm0.27$ & $31.71\pm0.85$ & $98.41\pm0.08$ \\
\bottomrule
\end{tabular}
\end{table}

Table~\ref{tab:budgetversions} reports mean retained budgets by benchmark and for the combined V1/V2 suite. The remaining distribution tables use the same 14-task V1/V2 macro average. Full requires all three axes to be maximal. Depth-bin labels are nominal; mean depth retention uses actual integer layers.

\begin{table}[htbp]
\centering\small
\caption{V1/V2 query shares by per-axis budget (\%), using 14-task macro means $\pm$ sample SD over three controller seeds. Depth bins are nominal (ColPali's 75\% bin uses 14/18 layers); token retention is shared by query and document sequences.}
\label{tab:axisdistributions}
\begin{tabular}{lrrrr}
\toprule
Backbone & 25\% & 50\% & 75\% & 100\% \\
\midrule
\rowcolor{datasetgray}\multicolumn{5}{l}{\textbf{Depth}} \\
ColPali & $0.00\pm0.00$ & $0.00\pm0.00$ & $1.29\pm0.45$ & $98.71\pm0.45$ \\
ColQwen2 & $0.00\pm0.00$ & $0.00\pm0.00$ & $0.01\pm0.01$ & $99.99\pm0.01$ \\
ColQwen2.5 & $0.00\pm0.00$ & $0.00\pm0.00$ & $0.01\pm0.00$ & $99.99\pm0.00$ \\
webAI & $0.20\pm0.16$ & $0.01\pm0.00$ & $99.47\pm0.15$ & $0.32\pm0.01$ \\
Nemotron & $0.00\pm0.00$ & $0.00\pm0.00$ & $6.08\pm1.06$ & $93.92\pm1.06$ \\
\rowcolor{datasetgray}\multicolumn{5}{l}{\textbf{Width}} \\
ColPali & $0.00\pm0.00$ & $15.66\pm0.56$ & $45.05\pm1.74$ & $39.29\pm1.31$ \\
ColQwen2 & $0.00\pm0.00$ & $0.11\pm0.01$ & $13.03\pm1.18$ & $86.87\pm1.19$ \\
ColQwen2.5 & $0.01\pm0.00$ & $0.08\pm0.11$ & $20.01\pm2.88$ & $79.91\pm2.95$ \\
webAI & $66.88\pm0.58$ & $27.37\pm0.05$ & $5.42\pm0.57$ & $0.33\pm0.07$ \\
Nemotron & $90.91\pm1.07$ & $0.20\pm0.09$ & $0.02\pm0.01$ & $8.87\pm1.17$ \\
\rowcolor{datasetgray}\multicolumn{5}{l}{\textbf{Tokens}} \\
ColPali & $0.00\pm0.00$ & $0.00\pm0.00$ & $36.50\pm2.44$ & $63.50\pm2.44$ \\
ColQwen2 & $0.00\pm0.00$ & $0.00\pm0.00$ & $1.41\pm0.43$ & $98.58\pm0.43$ \\
ColQwen2.5 & $0.00\pm0.00$ & $0.00\pm0.00$ & $0.60\pm0.11$ & $99.40\pm0.11$ \\
webAI & $0.00\pm0.00$ & $0.01\pm0.00$ & $8.58\pm2.17$ & $91.41\pm2.17$ \\
Nemotron & $0.00\pm0.00$ & $0.00\pm0.00$ & $6.37\pm0.32$ & $93.63\pm0.32$ \\
\bottomrule
\end{tabular}
\end{table}

\begin{table}[htbp]
\centering\small
\caption{Query shares by number of reduced axes (\%), mean $\pm$ sample SD over three controller seeds. Zero reduced axes denotes Native Full; displayed zeros may reflect rounding.}
\begin{tabular}{lrrrr}
\toprule
Backbone & Full (0 axes) & 1 axis & 2 axes & 3 axes \\
\midrule
ColPali & $28.54\pm2.52$ & $45.54\pm1.53$ & $24.80\pm1.16$ & $1.12\pm0.24$ \\
ColQwen2 & $85.56\pm1.45$ & $14.32\pm1.44$ & $0.11\pm0.03$ & $0.00\pm0.01$ \\
ColQwen2.5 & $79.37\pm2.90$ & $20.55\pm2.84$ & $0.07\pm0.06$ & $0.01\pm0.00$ \\
webAI & $0.02\pm0.01$ & $0.56\pm0.09$ & $90.86\pm2.06$ & $8.55\pm2.15$ \\
Nemotron & $8.86\pm1.16$ & $79.44\pm0.63$ & $10.94\pm0.98$ & $0.75\pm0.12$ \\
\bottomrule
\end{tabular}
\end{table}

The numerical companion includes all 64 joint route frequencies per backbone, their sample standard deviations, and per-seed values, together with per-task full-route use and nDCG@5.

\section{Random-assignment controls}
\label{app:assignment}
\begin{table}[htbp]
\centering\small
\caption{Matched random assignment on V1/V2 (seed 42, one permutation), preserving each task's route counts. Scores are 14-task macro nDCG@5 at the main preferences; $\Delta$ is Router minus matched random in percentage points.}
\label{tab:matched}
\begin{tabular}{lrrrr}
\toprule
Backbone & Router & Matched random & $\Delta$ (pp) & Full (\%) \\
\midrule
ColPali & 0.7485 & 0.7432 & +0.530 & 28.22 \\
ColQwen2 & 0.6978 & 0.6976 & +0.018 & 83.91 \\
ColQwen2.5 & 0.7205 & 0.7206 & -0.008 & 77.79 \\
\bottomrule
\end{tabular}
\end{table}

For matched random, we reset \texttt{default\_rng(42)} within each dataset and permute the router's route choices in saved query order, preserving the exact route multiset. This control covers ColPali, ColQwen2, and ColQwen2.5.

For the supplementary V1/V2 uniform-random control in Table~\ref{tab:controls}, we reset the same generator within each task and draw integer indices from all 64 routes, including full. The three backbones use the same normalized route indices. In expectation each route has probability $1/64$, giving a full-route share of 1.5625\%. The realized macro full share in this seed is 0.6791\%, and the realized average depth/width/token retentions are about 63.6\%/63.9\%/60.2\%. It is therefore not distribution-matched to the learned policy. Exact uniform expectations obtained by averaging every route utility are 0.5053, 0.4149, and 0.4547 nDCG@5; Table~\ref{tab:controls} reports the specified random realization instead. Both random controls use one realization.

\section{External calibration and preference sensitivity}
\label{app:calibration}
The main preferences listed in Section~\ref{sec:calibration} are selected with an absolute nDCG@5 tolerance of $\varepsilon_{\mathrm{cal}}=0.005$ (0.5 percentage points) and fixed across all three benchmarks and controller seeds.

Calibration uses 32 queries per V1/V2 task (448 total), with 16 disjoint queries per task (224 total) reserved for confirmation; each task retains its complete corpus. The coarse grids are $\{0,.01,.025,.05,.075,.1\}$ for four backbones and $\{0,.02,.05,.1,.15,.2\}$ for webAI, with midpoint refinement near the feasibility boundary. With representations and controllers fixed, selection averages tasks equally and then averages the three controllers, choosing the lowest mean structural penalty within the quality tolerance. Confirmation scores are reported in Appendix~\ref{app:confirmation}.

The full V1/V2 evaluation contains all 5,023 queries, including the calibration and confirmation subsets. The accompanying records provide the query manifests and calibration protocol.

\begin{figure}[htbp]
\centering
\includegraphics[width=\linewidth]{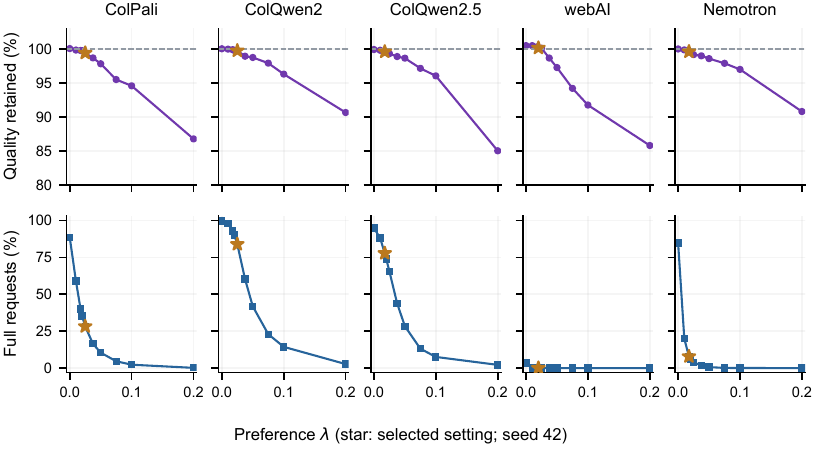}
\caption{\textbf{Supplementary resource-preference sensitivity (seed 42).} Top: nDCG@5 relative to Native Full. Bottom: full-route query share. Both use equal averaging over 14 V1/V2 tasks with all queries and complete corpora. Stars mark selected preferences.}
\label{fig:calibration}
\end{figure}

\begin{table}[htbp]
\centering\small
\caption{V1/V2 results at common resource preferences (seed 42), macro-averaged over 14 tasks. Full (\%) is the query share selecting Native Full; $\Delta$ is the nDCG@5 change from Native Full in percentage points.}
\label{tab:lambda}
\begin{tabular}{lrrrrrr}
\toprule
 & \multicolumn{3}{c}{$\lambda=0.1$} & \multicolumn{3}{c}{$\lambda=0.2$} \\
\cmidrule(lr){2-4}\cmidrule(lr){5-7}
Backbone & nDCG@5 & $\Delta$ (pp) & Full (\%) & nDCG@5 & $\Delta$ (pp) & Full (\%) \\
\midrule
ColPali & 0.7121 & -4.10 & 2.27 & 0.6533 & -9.98 & 0.14 \\
ColQwen2 & 0.6737 & -2.61 & 14.37 & 0.6342 & -6.56 & 2.71 \\
ColQwen2.5 & 0.6948 & -2.88 & 7.54 & 0.6151 & -10.86 & 2.16 \\
webAI & 0.7267 & -6.55 & 0.00 & 0.6795 & -11.27 & 0.00 \\
Nemotron & 0.7877 & -2.46 & 0.03 & 0.7373 & -7.49 & 0.00 \\
\bottomrule
\end{tabular}
\end{table}

The V1/V2 sweep evaluates $\{0,.01,.0175,.02,.025,.0375,.05,.075,.1,.2\}$ from cached seed-42 predictions; Industrial uses the separate 18-point grid in Appendix~\ref{app:variants}. Figure~\ref{fig:calibration} uses the 14-task macro average. In Figure~\ref{fig:preference}(b), each benchmark's task-macro nDCG@5 is instead normalized by its Native Full score, then the V1 and V2 ratios are averaged equally.

Figure~\ref{fig:preference}(a) uses the separate 1K computation workload and ratio-of-sums FLOPs aggregation specified in Appendix~\ref{app:flops}; panel (b) evaluates quality on the complete benchmark corpora.

\section{Route-conditioned FLOPs protocol}
\label{app:flops}
\paragraph{Workload and counting convention.}
The Industrial variant study in Figure~\ref{fig:variants} counts computation for its 283 queries and all 5,244 indexed pages, using the actual query and document lengths from each released processor. The hyperparameter sweep in Figure~\ref{fig:preference}(a) uses a common manifest of 1,000 text queries and 1,000 indexed pages for all five backbones, with controller seed 42. For this workload, query lengths are reproduced through the processors and document lengths come from the stored index. Both studies include Budget Allocator arithmetic only for learned policies; fixed and random choices incur none.

One multiply-add counts as two FLOPs. For a linear map with sequence length $S$, input width $d_i$, and output width $d_o$,
\[
F_{\mathrm{linear}}=2Sd_id_o.
\]
The dense-equivalent count for the two attention matrix products is
\[
F_{\mathbf Q\mathbf K^\top,\mathbf A\mathbf V}=4S^2n_hd_h,
\]
where $\mathbf Q$, $\mathbf K$, and $\mathbf V$ are the attention query, key, and value matrices, $\mathbf A$ is the attention-weight matrix, $n_h$ is the number of query heads, and $d_h$ is their width. Q/K/V/O projections use their actual dimensions, including grouped-query attention. This convention does not model causal-tile skipping or fused-kernel instruction counts. MLP projections, released unmerged LoRA operations, and output projection heads are also counted.

\paragraph{Matching and padding.}
For the 1K hyperparameter workload, all documents are padded to the global maximum retained length for the chosen route, even when dot products are processed in batches of 16 pages. With $N_{\mathrm{doc}}=1{,}000$ pages, $n_q$ valid output query tokens, and page lengths $n_j$, the counted dot-product work is
\[
F_{\mathrm{dot}}(q,r)=2N_{\mathrm{doc}}\lceil\tau n_q\rceil\Big[\max_j\lceil\tau n_j\rceil\Big]d.
\]
Industrial matching pads within each ordered block of eight documents. Let $\mathcal B$ partition its 5,244 pages into these blocks, with the last block retaining its actual size. Its dot-product count is
\[
F_{\mathrm{dot}}^{\mathrm{Industrial}}(q,r)
=2\lceil\tau n_q\rceil d\sum_{B\in\mathcal B}|B|\max_{j\in B}\lceil\tau n_j\rceil.
\]
The encoder still processes the complete input query length. Reducing output token ratio or embedding width does not reduce its layer computation. Intermediate projection heads produce their complete native output width before truncation; their arithmetic is counted at that width.

\paragraph{Model-specific operators.}
WebAI's hybrid layers include gating projections, convolution, and a PyTorch GatedDeltaNet implementation with chunk size 64. Counts include chunk padding, eight matrix products, and the explicit triangular recurrence in this implementation.

ColPali's final-layer query path computes a generation head of width 257,216 from hidden width 2,048, adding $2S\cdot2{,}048\cdot257{,}216$ operations. Its early-exit path avoids this head. The count includes this difference between the two implemented paths.

\paragraph{Router arithmetic and excluded work.}
Counted router work includes sparse feature projection, tree-output averaging, and utility arithmetic. Branch comparisons are not treated as FLOPs. Norms, activations, RoPE, most elementwise arithmetic and reductions, MaxSim max/sum, TopK, index-view normalization/conversion, tokenizer work, and data movement are excluded. The reported FLOPs quantify the listed arithmetic costs; end-to-end latency also depends on memory access and the excluded operations. Offline page encoding and index construction are outside the online-query count.

\paragraph{Aggregation.}
GFLOPs/query divides the sum of operation counts by $10^9$ times the number of queries. Retention is $100\sum_q F_{\mathrm{router}}(q)/\sum_q F_{\mathrm{full}}(q)$, relative to each backbone's Native Full.

\section{Theoretical Analysis}
\label{app:proofs}
Throughout this section, the finite route set $\R$, utility predictor $\hat u$, and structural penalty $c$ are fixed. The selection rule in Equation~\ref{eq:decision} may choose any maximizer.

\begin{proposition}[Resource-preference monotonicity]
For any query $q$ and $0\le\lambda_1<\lambda_2$, let $r_i\in\arg\max_{r\in\R}[\hat u(q,r)-\lambda_i c(r)]$. Then $c(r_2)\le c(r_1)$.
\end{proposition}
\begin{proof}
Write $\hat u_i=\hat u(q,r_i)$ and $c_i=c(r_i)$. Optimality at the two preferences gives
\[
\lambda_1(c_1-c_2)\le\hat u_1-\hat u_2\le\lambda_2(c_1-c_2).
\]
Hence $(\lambda_2-\lambda_1)(c_1-c_2)\ge0$. Since $\lambda_2>\lambda_1$, $c_2\le c_1$.
\end{proof}

For the following comparisons, fix $\lambda$ and a reference utility $u$, and write $v(q,r)=u(q,r)-\lambda c(r)$. The reference utility is the same for all compared decisions; the local supervision utility in Section~\ref{sec:router} and retrieval quality on a different corpus need not coincide.

\begin{proposition}[Utility error bounds decision regret]
Let $r^*\in\arg\max_{r\in\R}v(q,r)$ and $\hat r\in\arg\max_{r\in\R}[\hat u(q,r)-\lambda c(r)]$. If $\max_{r\in\R}|u(q,r)-\hat u(q,r)|\le\delta_u(q)$, then $v(q,r^*)-v(q,\hat r)\le2\delta_u(q)$.
\end{proposition}
\begin{proof}
Write $\hat v(q,r)=\hat u(q,r)-\lambda c(r)$. The regret equals
\[
\underbrace{v(q,r^*)-\hat v(q,r^*)}_{\le\delta_u(q)}+
\underbrace{\hat v(q,r^*)-\hat v(q,\hat r)}_{\le0}+
\underbrace{\hat v(q,\hat r)-v(q,\hat r)}_{\le\delta_u(q)}.
\]
The second term is nonpositive by the selection rule, giving the bound.
\end{proof}

\begin{proposition}[Oracle adaptive allocation versus fixed routes]
For integrable $v(q,r)$,
\begin{equation}
\E_q\max_r v(q,r)\ge\max_r\E_qv(q,r).
\label{eq:advantage}
\end{equation}
The inequality is strict if different routes uniquely maximize $v$ on query subsets of positive probability.
\end{proposition}
\begin{proof}
For every fixed $r_0$, $\max_rv(q,r)\ge v(q,r_0)$ pointwise. Taking expectations and then maximizing over $r_0$ proves the weak inequality. Under the stated condition, any fixed route fails to maximize on at least one positive-probability subset. Its pointwise gap is nonnegative everywhere and strictly positive on that subset, hence its expected gap is positive. Finiteness of the route set makes the best fixed-route gap positive as well.
\end{proof}

\paragraph{Non-dominated selections.}
For $\lambda>0$, a selected route is non-dominated in predicted utility and structural penalty: any route with no lower predicted utility and no higher penalty, with at least one strict improvement, would have a strictly larger selection objective. This observation concerns $(\hat u(q,r),c(r))$.
\end{document}